\documentclass[submission,copyright,creativecommons]{eptcs}
\providecommand{\event}{IWIGP 2012} 
\usepackage{ucs}
\usepackage[utf8x]{inputenc}
\usepackage[T1]{fontenc}
\usepackage{graphicx}
\usepackage{amssymb}
\usepackage{longtable}

\begin{document}

\newcounter{definition}
\newcommand{\definition}{\refstepcounter{definition} {\noindent \bf Definition \thedefinition: }}
\renewcommand{\thedefinition}{\arabic{definition}}

\newcounter{axiom}
\newcommand{\axiom}{\refstepcounter{axiom} {\vspace{0.5cm} \noindent \bf Axiom \theaxiom: }}
\renewcommand{\theaxiom}{\arabic{axiom}}

\newcounter{proposition}
\newcommand{\proposition}{\refstepcounter{proposition} {\vspace{0.5cm} \noindent \bf Proposition \theproposition: }}
\renewcommand{\theproposition}{\arabic{proposition}}

\newcounter{theorem}
\newcommand{\theorem}{\refstepcounter{theorem} {\vspace{0.5cm} \noindent \bf Theorem \thetheorem: }}
\renewcommand{\thetheorem}{\arabic{theorem}}

\newcounter{lemma}
\newcommand{\lemma}{\refstepcounter{proposition} {\vspace{0.5cm} \noindent \bf Lemma \thelemma: }}
\renewcommand{\thelemma}{\arabic{lemma}}

\newcounter{proof}
\newcommand{\proof}{\vspace{0.3cm} \noindent {\it Proof:} }

\title{Processes, Roles and Their Interactions}
\author{Johannes Reich, johannes.reich@sophoscape.de}
\def\titlerunning{Processes, Roles and Their Interactions}
\def\authorrunning{J. Reich}

\maketitle

\begin{abstract}
Taking an interaction network oriented perspective in informatics raises the challenge to describe deterministic finite systems which take part in networks of nondeterministic interactions. The traditional approach to describe processes as stepwise executable activities which are not based on the ordinarily nondeterministic interaction shows strong centralization tendencies. As suggested in this article, viewing processes and their interactions as complementary can circumvent these centralization tendencies.

The description of both, processes and their interactions is based on the same building blocks, namely finite input output automata (or transducers). Processes are viewed as finite systems that take part in multiple, ordinarily nondeterministic interactions. The interactions between processes are described as protocols. 

The effects of communication between processes as well as the necessary coordination of different interactions within a processes are both based on the restriction of the transition relation of product automata. The channel based outer coupling represents the causal relation between the output and the input of different systems. The coordination condition based inner coupling represents the causal relation between the input and output of a single system.

All steps are illustrated with the example of a network of resource administration processes which is supposed to provide requesting user processes exclusive access to a single resource. 
\end{abstract}


%
\section{Introduction}
%


We are living in an open world of interaction relations. As an example, Fig. \ref{fig_network_of_relations} sketches a cutout of a network of business relations between a buyer, a seller, its stock, a post and a bank. Intuitively, these networks have two key features. First, they represent peer relations in the sense that none of the participants is in total control of all the other participants. The participant's actions are, in general, not determined by their interactions. Second, they are open in the sense that we will never be able to describe them completely. Each participant will also interact in many other roles in other networks.

\begin{figure}[ht]
\begin{center}
\includegraphics[width=7cm]{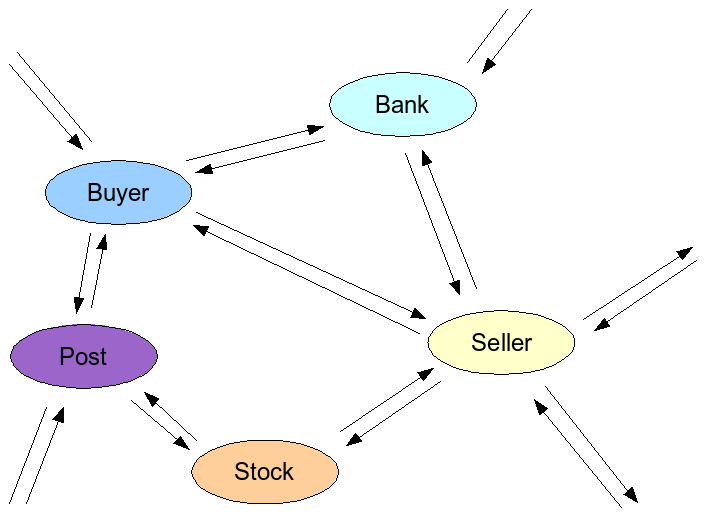}
\end{center}
\caption[]
{\label{fig_network_of_relations}A cutout of an open business network.}
\end{figure}

In informatics, the interacting systems are often called processes. So, taking a process or interaction network oriented perspective in informatics raises the challenge to describe deterministic finite systems which take part in networks of nondeterministic interactions. The traditional approach to model processes as stepwise executable activities leads to strong centralization tendencies (see \cite{reich2011, DBLP:journals/tse/DesaiMCS05}).
The contribution of this article is a formalism based on finite input output automata to describe processes and their interactions in a complementary way avoiding these centralization tendencies. We present two coupling mechanisms for system parts: the outer or interaction coupling, described by protocols and the inner or action coupling, leading to processes. The duality of system specifications from a process perspective and from an interaction perspective actually has been recognized as early as 1976 by Chris A. Vissers \cite{Vissers1976}.

The rest of this article is structured as follows. Next, a concrete example of such an open interaction network is introduced. In section \ref{s_building_blocks} the technical building blocks, the finite input output automata, are presented. Section \ref{s_channelbr_protocols} and \ref{s_condbr_processes} define the outer and inner coupling mechanisms. In section \ref{s_related_work} related work is presented and in section \ref{s_discussion} the results are discussed.

%
\subsection{Preliminaries}
%
Throughout this article, elements and functions are denoted by small letters, sets and relations by large letters and mathematical structures by large calligraphic letters. 
The components of a structure may be denoted by the structure's symbol or - in case of enumerated structures - index as subscript. The subscript is dropped if it is clear to which structure a component belongs.

For any character set or alphabet $A$, $A^\epsilon := A\cup\{\epsilon\}$ with $\epsilon$ is the empty character. For state value sets $Q$, $Q^\epsilon := Q\cup\{\epsilon\}$ with $\epsilon$ is the undefined value.   If either a character or state value set $A = A_1 \times \dots \times A_n$ is a Cartesian product then $A^\epsilon = A_1^\epsilon \times \dots \times A_n^\epsilon$. Additionally,  a state vector $(p_1,  \dots, p_n)$ where $p_k$ belongs to structure ${\cal A}_k$ is written as $\vec{p}$ and the change of this vector in a position $k$ from $p$ to $q$ is written as $\vec{p} \left[\frac{q}{p}, k\right]$. An n-dimensional vector of characters with the $k$-th component $v$ and the rest $\epsilon$ is written as $\vec{\epsilon}[v,k]=(\epsilon_1, \dots, \epsilon_{k-1}, v, \epsilon_{k+1}, \dots, \epsilon_n)$. 

In the main article, I will talk about processes and their parts in a relational sense, which means that I actually talk about them at the level of their specification. The relation to the functional level of (finite) systems is shown in the appendix.

The part of a process that is directly related to an interaction is called a {\it role}.

%
\subsection{Example}
%
In Fig. \ref{fig_resource_manager_as_ring_II} the open network of resource administration processes is shown. In its core, it consists of two different interactions. The first interaction is about the classical problem of mutual exclusion (e.g. \cite{Lynch1996}). Some process requests exclusive access to a resource and a resource administration process admits it.
The second interaction concerns the coordination of the different resource administration processes for administering the exclusive access to the resource, where I chose the token ring protocol. 

\begin{figure}[ht]
\begin{center}
\includegraphics[width=10cm]{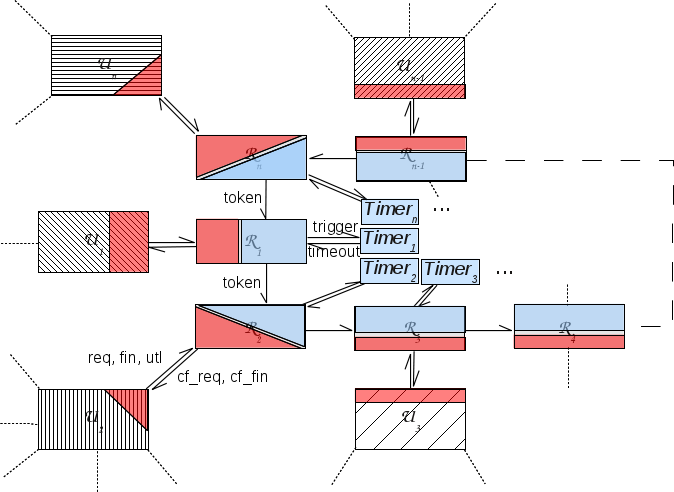}
\end{center}
\caption[]
{\label{fig_resource_manager_as_ring_II} A network of interacting processes where $n$ resource administrator processes ${\cal R}_i$ interact in a ring topology with timers and resource demanding processes ${\cal U}_i$, which themselves interact with an unknown number of other processes. The different colors mark the ''roles'' of the processes, that is those process parts that are directly related to the interactions. Red marks the protocol of mutual exclusion, blue marks the token ring protocol.}
\end{figure}

\subsubsection{First Interaction: Protocol of Mutual Exclusion} \label{sss_first_interaction}

\begin{figure}[ht]
\begin{center}
\includegraphics[width=8cm]{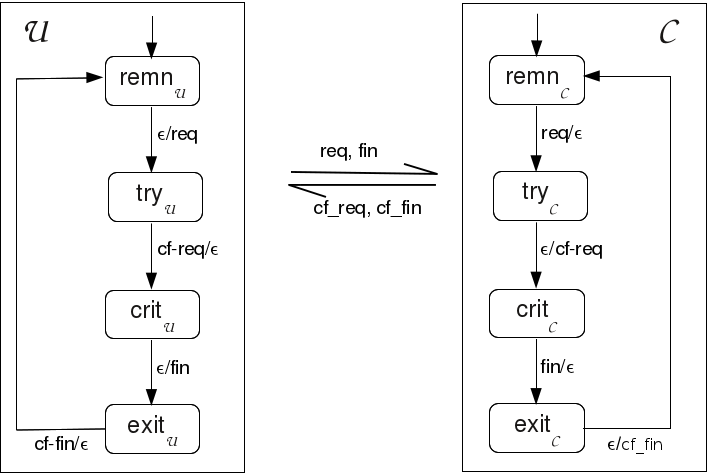}
\end{center}
\caption[]
{\label{fig_mutual_exclusion} \rm State diagram of the protocol of mutual exclusion where a resource administrator process in its role b${\cal C}$ manages the exclusive usage of a single resource for the user process in its role ${\cal U}$.}
\end{figure}

The first interaction with its two roles ${\cal U}$ and ${\cal C}$ is displayed in Fig. \ref{fig_mutual_exclusion}. ${\cal U}$ has the state values $remn_{\cal U}$ (for remainder region), $try_{\cal U}$ (for trying region), $crit_{\cal U}$ (for critical region) and $exit_{\cal U}$ (for exit region), while ${\cal C}$ has the state values $remn_{\cal C}$, $try_{\cal C}$, $crit_{\cal C}$ and $exit_{\cal C}$.

Initially, ${\cal U}$ and ${\cal C}$ are in the state $(remn_{\cal U}, remn_{\cal C})$. If ${\cal U}$ requests the resource, it notifies ${\cal C}$ with a $req$-character (for request) and transits into its $try_{\cal U}$ state. If the resource is free for use, ${\cal C}$ confirms the request with a {\it cf-req}-character, transits into its $crit_{\cal C}$ state and now provides the resource for usage. ${\cal U}$ receives the {\it cf-req}-character, transits to its $crit_{\cal U}$ state and now could use the resource (omitted). ${\cal U}$ releases the resource with sending a $fin$-character (for finalize), which is confirmed by ${\cal C}$ with sending a {\it cf-fin}-character back to ${\cal C}$, so that at the end of the cycle, both, ${\cal U}$ as well as ${\cal C}$ are again in their $remn$ state.

\subsubsection{Second Interaction: Token Ring Protocol}

The token ring protocol is illustrated in Fig. \ref{fig_mutual_exclusion_token} and consists of a quadruple interaction between a process, its two neighbor processes, and a timer for each process.
Each ${\cal R}_i$ can take one out of three state values $abst$ (for absent), $avlb$ (for available) and $interm$ (for intermediate). Each timer can take either $wait$ or $triggered$.  The protocol has to ensure that one ${\cal R}_i$ is not in $abst$ where it is allowed to provide access to the resource.
Initially, ${\cal R}_1$ is initialized with value $avlb$ and $timer_1$ with $triggered$ while all other processes are initialized with $abst$ and all other timers with $wait$.

${\cal R}_i$ waits in the state $abst$ for reception of the token from ${\cal R}_{i-1}$ (or in case of $i=1$, ${\cal R}_1$ waits for ${\cal R}_n$) to transit into state $avlb$ in case of its reception and to trigger its timer. It rests there until it receives the timeout which results in a transition to $interm$. Finally, ${\cal R}_i$ hands over the token to ${\cal R}_{i+1}$ (or in case of $i=n$, ${\cal R}_n$ hands over to ${\cal R}_1$) and transits back to $abst$.

%
\section{The Building Blocks: Finite Input/Output Automata}\label{s_building_blocks}
%



Finite input output automata will later be used to describe both, processes and their interactions.

\begin{figure}[ht]
\begin{center}
\includegraphics[width=10cm]{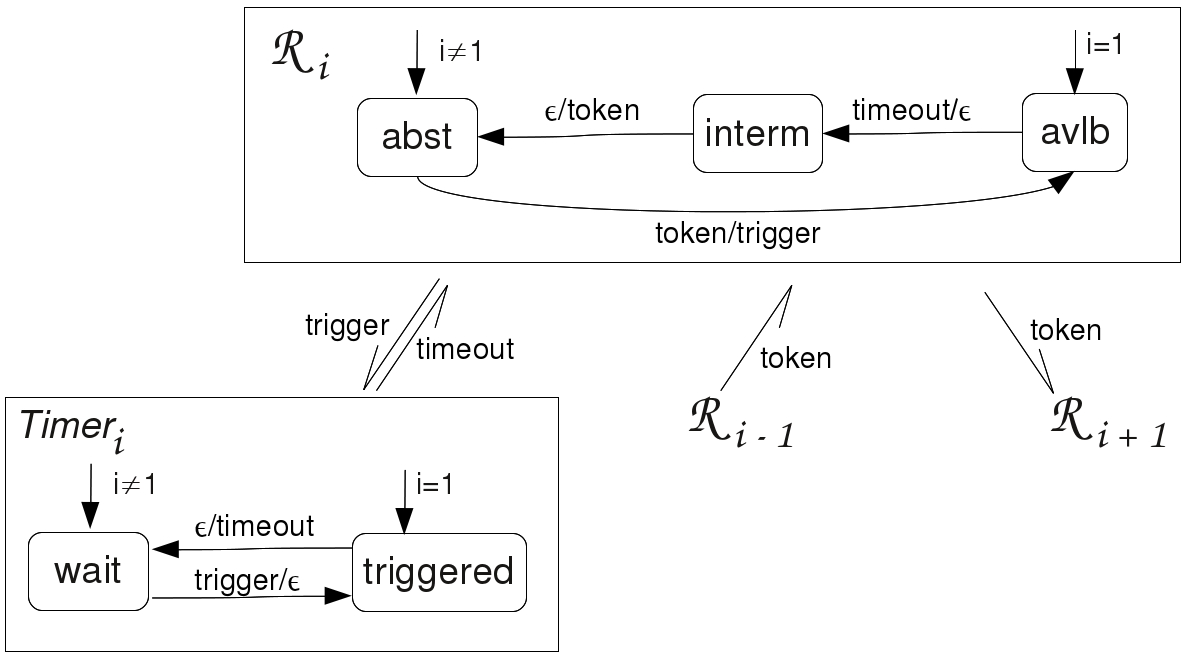}
\end{center}
\caption[]
{\label{fig_mutual_exclusion_token}\rm State diagram of the token ring protocol for managing exclusive access to a single resource by a resource administrator process ${\cal R}_i$.}
\end{figure}

\begin{definition} \label{def_NFIOA} A {\it nondeterministic finite I/O automaton (NFIOA)} is a tuple ${\cal A} = (Q, I, O, q_0, Acc, \Delta)$.  $Q$ is the non-empty finite set of state values, $I$ and $O$ are the finite, possibly empty input and output alphabets. The elements of $Q$, $I$ or $O$ are usually assumed to be vectors. For each element of $I$ and $O$ at most one component is unequal to $\epsilon$. $q_0$ is the initial state value, $Acc$ is the acceptance component and $\Delta \subseteq Q\times Q\times I^\epsilon\times O^\epsilon$ is the transition relation.
\end{definition}

The acceptance component $Acc$ represents the information needed to express the acceptance condition. It depends on the computational model of ''acceptance''. For a finite computation $Acc$ is a set of final state values. For an infinite computation a representative acceptance condition is to accept all runs in which the set of infinitely often occurring state values is an element of the acceptance component (Muller-acceptance), that is $Acc \subseteq 2^{Q}$ \cite{DBLP:conf/dagstuhl/Farwer01}.

Transitions with an $\epsilon$-character as input are called {\it spontaneous}. In case that for each $(p, i) \in Q\times I^\epsilon$ there is at most one transition $(p, q, i, o) \in \Delta$ then $\Delta$ specifies a function $\delta: Q\times I^\epsilon \rightarrow Q\times O^\epsilon$ with $(q,o) = \delta(p, i)$. If also no spontaneous transition exists, we have a deterministic automaton or {\it DFIOA} (a Mealy automaton \cite{Mealy1955}).



%
\subsection{Product Automata}
%
The precondition of any coupling of NFIOAs is to view them as a single automaton where all to-be-coupled NFIOAs are viewed together as a product automaton, stepping synchronously and where all acceptance conditions are to be fulfilled simultaneously.

\begin{definition}\label{def_weakly_synchronized_product} 
The {\it weakly synchronized product} of a set of $n$ NFIOAs ${\cal A}_k$ is defined by NFIOA ${\cal B} = (Q, I, O, \vec{q}_0, Acc, \Delta)$, with $Q_{\cal B} = \mbox{\Large $\times$} Q_k$, $I_{\cal B} = \mbox{\Large $\times$} I_k$, $O_{\cal B} = \mbox{\Large $\times$} O_k$, ${{\stackrel{\rightarrow}{q}}_0}_{\cal B} = ({q_0}_1, \dots, {q_0}_n)$,  the common acceptance component represents the logical conjunction of the individual components, symbolized as $Acc_{\cal B} = \bigwedge Acc_k$, $\Delta_{\cal B} := \{(\vec{p}, \vec{q}, \vec{i}, \vec{o})|$ the components of $\vec{p}$ are reachable states of the ${\cal A}_i$ and ${\cal A}_k$ provides a transition $(p_k, q_k, i_k, o_k)$ with $\vec{q} = \vec{p} \left[\frac{q_k}{p_k}, k\right]$ and $\vec{i} =\epsilon[i_k, k]$ and $\vec{o} = \epsilon[o_k, k]$ $\}$. I also write ${\cal B} = \bigotimes_{i=1}^n {\cal A}_i$.
\end{definition}

As each step of the product automaton is achieved by only one component automaton, the input and output components of the product automaton again differ from epsilon at most in one component.

%
\section{Channel Based Restriction and Protocols} \label{s_channelbr_protocols}
%
\begin{figure}[ht]
\begin{center}
\includegraphics[width=12cm]{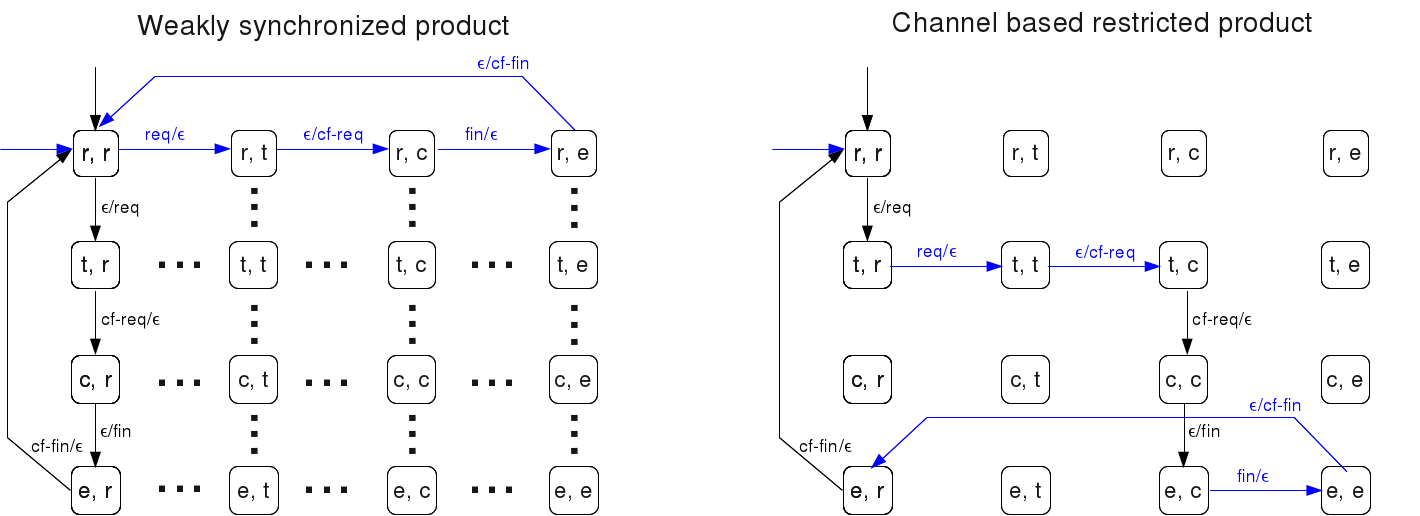}
\end{center}
\caption[]
{\label{fig_mutual_exclusion_synchronized_product} Left, you see the state diagram of the weakly synchronized product of both roles of the protocol of mutual exclusion. Right, you see the channel based restricted product. The names of the state values have been abbreviated.}
\end{figure}

Connecting two systems with a channel imposes certain restrictions on the transition relation of their weakly synchronized product automaton as it requires the receiving system to transit only after the sending system and with the input of the sending system's output. It thereby transforms the weakly synchronized product automaton into a stronger synchronized product automaton with a restricted transition relation. 

As an example, Fig. \ref{fig_mutual_exclusion_synchronized_product} shows the state diagram of both, the weakly synchronized product and the channel based restricted product of the protocol of mutual exclusion.

\begin{definition} \label{def_outer_coupling}
Be ${\cal T}$ an NFIOA where $O_k$    is the $k$-th component of the output alphabet and $I_l$ is the $l$-th component of the input alphabet and $O_k \subseteq I_l$. 
The pair $(k,l)$ is called a {\it channel} if it parametrizes the {\it  channel based restriction (cbr) operator} in the following way leading to the channel based restricted automaton ${\cal P} = cbr_{(k,l)}({\cal T})$.

\begin{enumerate}
\item $Q_{\cal P} = Q_{\cal T}$, $I_{\cal P} = I_{\cal T}$, $O_{\cal P} = O_{\cal T}$, $Acc_{\cal P} = Acc_{\cal T}$, ${q_0}_{\cal P} = {q_0}_{\cal T}$. 
\item Be $p\in Q_{\cal P}$ a reachable state of ${\cal P}$ (and thereby also of ${\cal T}$) which was reached by a transition with an output component $\epsilon[o,k] \in O_{\cal P}$ (i.e. $o\in{O_k}$). Then all transitions of $\Delta_{\cal T}$ which start from $p$ and have $\epsilon[o,l] \in I_{\cal P}$ (i.e. $o\in {I_l}$) as input character are also elements of $\Delta_{\cal P}$. I call the transition providing the character {\it sending} transition, the ensuing transitions {\it receiving} transitions and such a state $p$ {\it excited}.
\item Be $p\in Q_{\cal P}$ either the initial state or a reachable state of ${\cal P}$ which was reached by a transition with a character $o\in O_{\cal P}$ with $o_k = \epsilon$ that is no output on the channel. Then all spontaneous transitions of ${\cal T}$ which start from $p$ and all transitions  of ${\cal T}$ which have an input not associated with the channel are elements of $\Delta_{\cal P}$. I call such a state $p$ {\it relaxed}\footnote{Please note that being excited or relaxed is per se not a static but a dynamic attribute of a state as it depends on how it was entered.}.
\end{enumerate}

Input and output components which are not associated with a channel are called {\it open}. For an illustration of the resulting transition classes, see Tab. \ref{tab_vertauschung_kanaele_Ia}
\end{definition}

\begin{table}[ht]
\begin{center}
\begin{tabular}{c|c|c|c}
Type of start state & Input & Output & Type of target state\\
\hline
relaxed & $\epsilon$ & $\epsilon$ & relaxed\\
relaxed & $\epsilon$ & $\epsilon[\beta, k]$ & excited\\
relaxed & $\epsilon$ & $\epsilon[\beta, i\neq k]$ & relaxed\\
relaxed & $\epsilon[\alpha, i\neq l]$ & $\epsilon$ & relaxed\\
relaxed & $\epsilon[\alpha, i\neq l]$ & $\epsilon[\beta, k]$ & excited\\
relaxed & $\epsilon[\alpha, i\neq l]$ & $\epsilon[\beta, i\neq k]$ & relaxed\\
excited & $\epsilon[\alpha, l]$ & $\epsilon$ & relaxed\\
excited & $\epsilon[\alpha, l]$ & $\epsilon[\beta, k]$ & excited\\
excited & $\epsilon[\alpha, l]$ & $\epsilon[\beta, i\neq k]$ & relaxed
\end{tabular}
\end{center}
\caption[]{\label{tab_vertauschung_kanaele_Ia} This table shows the possible transition classes for channel based restrictions. $\epsilon[\alpha, i\neq l]$ means a class of n-dimensional character vectors with only some fix $v_{i\neq k} = \alpha \neq \epsilon$ where $\alpha$ is a character variable and not an individual character. 

A relaxed state allows for 6 possible transition classes. As input there is either $\epsilon$ for a spontaneous transition or there is a character not associated with the channel $(k,l)$, namely $\epsilon[\alpha, i\neq l]$. As output there is either a character with a non vanishing component $k$ on the channel, leading to a new excited state, or a character with a non vanishing component $i\neq k$ somewhere else, or no output at all.
 
An excited state allows for three possible transition classes as there must be a channel associated input character $\alpha$. The output is the same as for the relaxed state.} 
\end{table}

\begin{proposition}
Be ${\cal T}$ an NFIOA whose input and output alphabet are at least 2 dimensional and two channels $(k,l)$ and $(m,n)$ with $k\neq m$ and $l\neq n$. Then the $cbr$-operator for either channel commute: $cbr_{k,l}(cbr_{m,n}({\cal T})) = cbr_{m,n}(cbr_{k,l}({\cal T}))$. Thus, the $cbr$-operator can also be indexed with a set of channels.
\end{proposition}

\begin{proof}
To prove the proposition, it suffice to show that the rules, prescribing the proper transitions for the combined operator do not depend on the sequence of application of the operators for both channels. This is illustrated in Tab. \ref{tab_vertauschung_kanaele_Ib}
\end{proof}

\begin{table}[ht]
\begin{center}
\begin{tabular}{c|c|c|c}
Type of start state & Input & Output & Type of target state\\
\hline
relaxed & $\epsilon$ & $\epsilon$ & relaxed\\
relaxed & $\epsilon$ & $\epsilon[\beta, i=k \,\vee\, i=m]$ & excited\\
relaxed & $\epsilon$ & $\epsilon[\beta, i\neq k \,\&\, i\neq m]$ & relaxed\\
relaxed & $\epsilon[\alpha, i\neq l \,\&\, i\neq n]$ & $\epsilon$ & relaxed\\
relaxed & $\epsilon[\alpha, i\neq l \,\&\, i\neq n]$ & $\epsilon[\beta, i=k \,\vee\, i=m]$ & excited\\
relaxed & $\epsilon[\alpha, i\neq l \,\&\, i\neq n]$ & $\epsilon[\beta, i\neq k \,\&\, i\neq m]$ & relaxed\\
excited & $\epsilon[\alpha, i=l \,\vee\, i=n]$ & $\epsilon$ & relaxed\\
excited & $\epsilon[\alpha, i=l \,\vee\, i=n]$ & $\epsilon[\beta, i=k \,\vee\, i=m]$ & excited\\
excited & $\epsilon[\alpha, i=l \,\vee\, i=n]$ & $\epsilon[\beta, i\neq k \,\&\, i\neq m]$ & relaxed
\end{tabular}
\end{center}
\caption[]{\label{tab_vertauschung_kanaele_Ib} This table illustrates the restrictions of 2 consecutive channel operators onto a transition relation. The terminology is the same as in Tab. \ref{tab_vertauschung_kanaele_Ia}

A relaxed state allows for 6 different transition classes.  Either there is a spontaneous transition or there is an input $\alpha$ neither associated with channel $(k,l)$ nor with $(m,n)$. As output there is either a character with a non vanishing component $k$ or $m$ on a channel, each leading to a new excited state, or a character with a non vanishing component $i$ neither equaling $k$ nor $m$, or no output at all, both leading to a new relaxed state. 

An exited state allows for 3 possible transition classes. There must be a channel associated input character $\alpha$, either from $(k,l)$ or $(m,n)$. The output is the same as for the relaxed state.

These rules obviously do not depend on the sequence of the application of the operators for both channels.}
\end{table}


\begin{proposition} \label{prop_restr_commute}
Be ${\cal A}$ and ${\cal B}$ two NFIOAs and $restr$ an operator that restricts the transition relation of ${\cal B}$. Then the following relation holds: $restr({\cal A} \otimes {\cal B}) = {\cal A} \otimes restr({\cal B})$.
\end{proposition}

The proof is trivial.

\begin{definition}
A {\it run} as a sequence of states has to be calculated according to the following rules:
\begin{enumerate}
\item Be $p$ a relaxed state, then every spontaneous transition or every transition which have an open input (not associated with a channel) can be taken.
\item Be $p$ an excited state, then any ensuing transition must process this character.
\end{enumerate}
\end{definition}

It is not a priori clear that this procedure creates something meaningful. A necessary condition obviously is that there has to be a receiving transition for each sending transition for each channel and that the acceptance condition can still be met. I therefore define in accordance with the protocol literature (e.g. \cite{Brand1983}):

\begin{definition}
A channel based restricted automaton is called {\it well formed} if for every channel mediated transition which sends a character (different from $\epsilon$) there exists an induced transition to process it. 
\end{definition}
                                                                                                                                                                                                                                                                                                                                                                                                                                                                   
\begin{definition}
A well formed channel based restricted automaton is called {\it consistent} if for each reachable state value either the acceptance condition is met or there is at least one continuation such that the acceptance condition can be met. 
\end{definition}

As a simple consequence, in a consistent channel based restricted automaton all continuations allow to meet the acceptance condition. Assuming the contrary that one continuation leads to a reachable state where neither the acceptance condition is already met nor any such continuation exists is a direct contradiction to the definition. 

%
\subsection{Protocols}
%
I would like to focus on the special case where several NFIOAs interact completely via a known set of channels, that is with no open input or output components. 

\begin{definition}
A {\it protocol} is a channel based restricted product automaton with no open input or output components. I call the individual factor automata {\it roles}.
\end{definition}

As was already mentioned in \cite{Reich2010} the definition of ''consistent'' directly entails that a consistent protocol neither has deadlocks in the sense that there is a single reachable state without any continuation such that the acceptance condition holds nor livelocks, characterized by a set of periodical reached states without such a continuation. 
  
\begin{proposition} \label{prop_product_protocol}
A weakly synchronized product of two protocols is again a protocol.
\end{proposition}

\begin{proof} We have to show that for two protocols ${\cal P}_1 = cbr_{C_1}\bigotimes {\cal A}_i$ and ${\cal P}_2 = cbr_{C_2}\bigotimes {\cal B}_i$ with two sets of channels $C_1$ and $C_2$ the following equation holds: 

${\cal P}_1 \otimes {\cal P}_2 = (cbr_{C_1}\bigotimes {\cal A}_i) \otimes (cbr_{C_2}\bigotimes {\cal B}_i) = cbr_{C_1 \cup C_2} ((\bigotimes {\cal A}_i) \otimes (\bigotimes {\cal B}_i)$. 

\noindent This is obviously the case, as the restriction relates only either to the weakly synchronized ${\cal A}_i$ or ${\cal B}_i$ but never simultaneously to both.  Therefore it gives the same result, if we weakly synchronize two channel based restricted weakly synchronized products or create the channel based restriction over the set union of channels of the weakly synchronized products of all automata. 
\end{proof}

%
\section{Condition Based Restriction and Processes} \label{s_condbr_processes}
%
The second mechanism to couple NFIOAs represents the causal relation between the input and output of a single system as an inner coupling determined by coordination rules. The goal is to restrict the transition set of the weakly synchronized product automaton such that at least a quasi-determinism of the formerly nondeterministic transition set is achieved in a sense that from each reachable state there is at most one transition for each input character, including the empty character, while the factor NFIOAs (the roles) still can be regained by projection. 

In Fig \ref{fig_mutual_exclusion_man_in_the_middle} and \ref{fig_mutual_exclusion_synchronized_product_inner} an example of such an inner coupling for the mutual exclusion protocol is illustrated where a process realizes a ''man in the middle''. It just pass-through the incoming requests and thereby combines both protocol roles complementary to example \ref{sss_first_interaction}.

\begin{figure}[ht]
\begin{center}
\includegraphics[width=8cm]{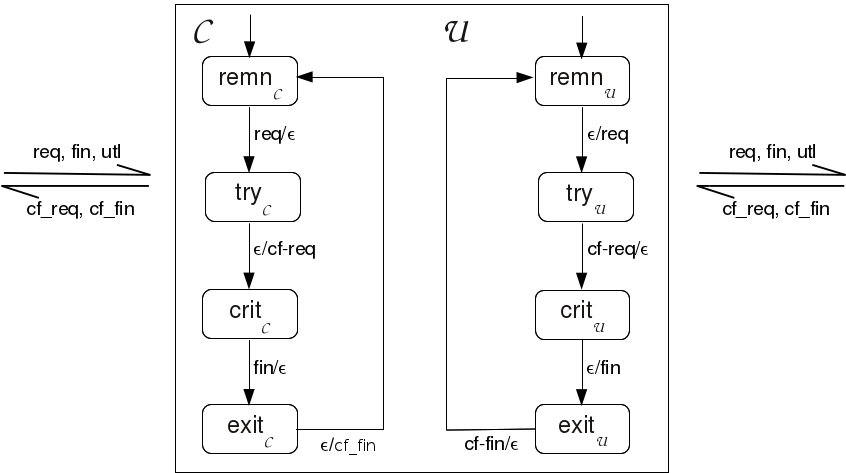}
\end{center}
\caption[]
{\label{fig_mutual_exclusion_man_in_the_middle} State diagram of the two (still independent) roles of a ''man in the middle'' process within the mutual exclusion protocol.}
\end{figure}

\begin{figure}[ht]
\begin{center}
\includegraphics[width=12cm]{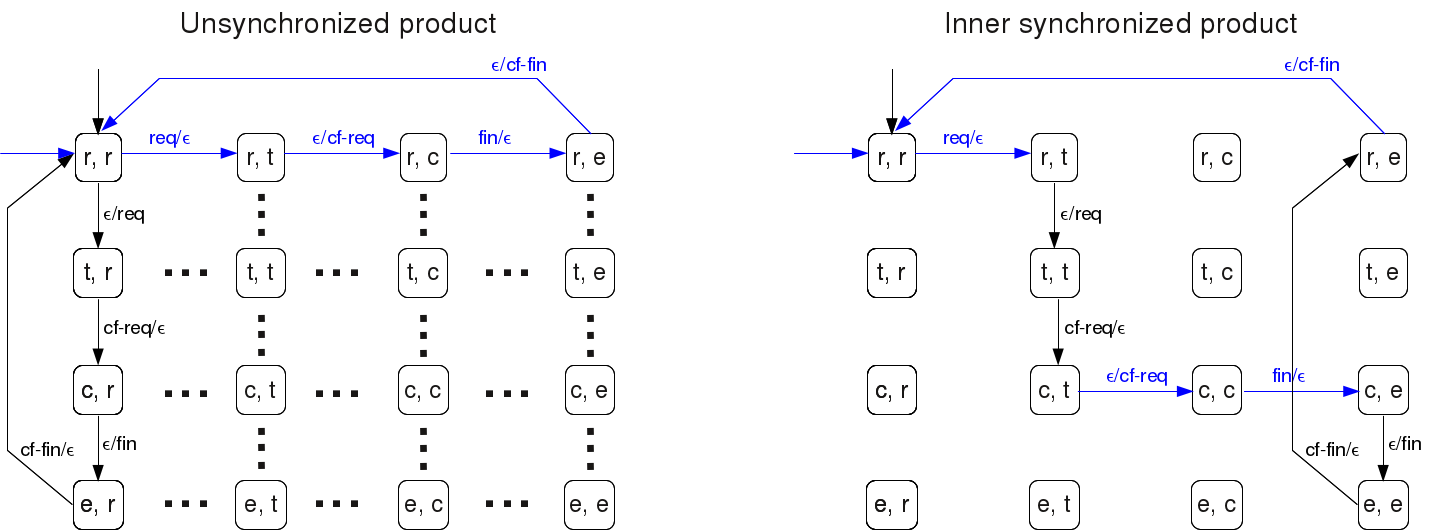}
\end{center}
\caption[]
{\label{fig_mutual_exclusion_synchronized_product_inner} State diagram of the two roles of the mutual exclusion protocol, acting as a ''man in the middle''. Left, you see the weakly synchronized product of both roles. Right, you see the condition based restricted product.} \end{figure}

\begin{definition}
A {\it quasi deterministic} NFIOA has at most one transition for each input sign, including the empty character, from any reachable state. 
\end{definition}

\begin{definition} \label{def_projected_automaton}
Be ${\cal B}$ an NFIOA and $\pi = (\pi_Q, \pi_I, \pi_O)$ with $\pi_Q: Q_{\cal B} \rightarrow Q_{\cal B}$, $\pi_I: I_{\cal B} \rightarrow I_{\cal B}$, and $\pi_O: O_{\cal B} \rightarrow O_{\cal B}$ a projection function\footnote{A projection function $\pi$ has the special property that $\pi = \pi \circ \pi$.}.
Then the projected automaton ${\cal A} = \pi({\cal B})$ is given by $Q_{\cal A} = \pi_Q(Q_{\cal B})$, $I_{\cal A} = \pi_I(I_{\cal B})$, $O_{\cal A} = \pi_O(O_{\cal B})$, ${q_0}_{\cal A} = \pi_Q({q_0}_{\cal B})$, $Acc_{\cal A} = \pi_Q(Acc_{\cal B})$, $\Delta_{\cal A} = \{(p',q',i',o')| (p, q, i, o) \in \Delta_{\cal B} \,\mbox{and}\, p'=\pi_Q(p), q'=\pi_Q(q), i'=\pi_I(i), o'=\pi_O(o)\}$.
\end{definition}

The condition based restriction is supposed to eliminate all transitions from non-reachable states as well as all transitions for which a given condition is true. 

\begin{definition}
Be ${\cal T}$ an NFIOA. A set of conditions $E$ parametrizes the {\it condition based restriction ($cond_E$) operator} such that ${\cal P} = cond_E({\cal T})$ with $Q_{\cal P} = Q_{\cal T}$, $I_{\cal P} = I_{\cal T}$, $O_{\cal P} = O_{\cal T}$, $Acc_{\cal P} = Acc_{\cal T}$, and $\Delta_{\cal P} = \{(p,q,i,o)\in\Delta_{\cal T}| p \,\mbox{is a reachable state of ${\cal T}$}\,\\ \mbox{and}\, e(p,q,i,o)\in E \, \mbox{is false}\}$.
\end{definition}

The interesting property of the condition based restriction is, that it may not affect the projection:

\begin{definition}
A projection $\pi$ of an NFIOA ${\cal T}$ is called {\it unaffected} by a condition based restriction $cond_E$ if $\pi({\cal T}) = \pi(cond_E({\cal T}))$. 
\end{definition}
                                                                                                                                                                                                                                                                                                                                                                                                                                                 
\begin{definition}
A condition restricted automaton is called {\it consistent} if if for each reachable state value either the acceptance condition is met or there is at least one continuation such that the acceptance condition can be met.
\end{definition}

%
\subsection{Separating inner and outer coupling}
%
Now, I would like to justify that looking at outer and inner coupling of roles can be separated. I first state that

\begin{proposition} \label{prop_channel_condition_commute}
Channel and condition based restriction commute.
\end{proposition}

\begin{proof}
To prove the proposition we have to show again that the rules to restrict the transition relation do not depend on the sequence of application of both operators. 
Starting from the initial state, the set of reachable states is build up inductively. For the initial state, the set of transition to be eliminated is just the union of the to-be-eliminated transitions from channel and condition based restriction. The thereby determined set of reachable next level nodes is therefore independent of the sequence of application of both operators. The same argument applies inductively for the transitions starting from this set as well as for any further set of reachable nodes. 
So in fact, as both operator definitions relate to reachable states, the rules for both, excited states and relaxed states of definition \ref{def_outer_coupling} become only modified in the sense that it has to be checked whether the outgoing transitions are to be additionally excluded by the condition based restriction.
\end{proof}

With this in mind, let us assume to have two independent protocols ${\cal P}_1 = cbr_{C_1}(\bigotimes {\cal A}_i)$ and ${\cal P}_2 = cbr_{C_2}(\bigotimes {\cal B}_i)$ where the roles ${\cal A}_n$ and ${\cal B}_1$ actually relate to the same process. The effective inner coupling between these two roles is expressed by a condition based restriction operator $cond_E$ relating only to ${\cal A}_n$ and ${\cal B}_1$. We then have: 
\begin{eqnarray*}
\lefteqn{cond_E({\cal P}_1 \otimes {\cal P}_2)}\\
 & \stackrel{P.\ref{prop_product_protocol}}{=} & cond_E(cbr_{C_1 \cup C_2}((\bigotimes {\cal A}_i) \otimes (\bigotimes{\cal B}_i )))\\
 & \stackrel{P. \ref{prop_restr_commute} \& \ref{prop_channel_condition_commute} }{=} & cbr_{C_1 \cup C_2}((\bigotimes_{i<n} {\cal A}_i) \otimes cond_E({\cal A}_n \otimes {\cal B}_1) \otimes (\bigotimes_{i>1}{\cal B}_i )))
\end{eqnarray*}

So, applying the condition based restriction to some large protocol, involving many different roles boils down to applying it only to the roles that are actually coordinated within a given process. 
In the general case, where the condition based restriction relates to arbitrary NFIOA parts of the protocols, these NFIOA would have to be permuted to allow for this separation.

%
\subsection{Processes}
%
Now, what is a process? As a first try, we could define a process within the presented framework as a product automaton from at least two roles of different protocols that are restricted by some coordination conditions. Thereby, a process coordinates at least two different interactions. 

Fig. \ref{fig_mutual_exclusion_token_alg_sync1a} illustrates the example of the introduction. It shows the weakly synchronized product automaton of the two different roles of a resource administration process together with the to-be-eliminated transitions according to the following coordination conditions:

\begin{enumerate}
\item Being in state value $(*, abst)$, no spontaneous transition to state value $(crit, *)$ is allowed. That is, the token is needed for any entry into the critical region.
\item Being in state value $(try, *)$, no spontaneous transition to state value $(*, abst)$ is allowed. That is, after an incoming request, it is not allowed to dispose the token.
\item Being in state value $(crit, *)$, no spontaneous transition to state value $(*, abst)$ is allowed. That is, after an incoming request has been acknowledged, it is not allowed to dispose the token.
\item Being in state value $(*, interm)$, no spontaneous transition to state value $(remn, *)$ is allowed. That is, before sending a confirmation for finalization, the token has to be disposed.
\end{enumerate}

\begin{figure}[ht]
\begin{center}
\includegraphics[width=10cm]{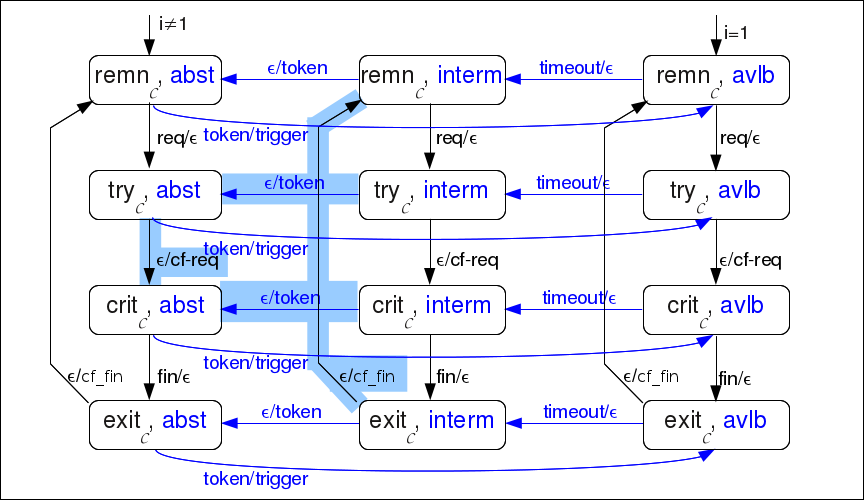}
\end{center}
\caption[]
{\label{fig_mutual_exclusion_token_alg_sync1a} The weakly synchronized product automaton between the two different roles of a resource administration process is shown. The transitions which are eliminated for coordination are marked.}
\end{figure}

\subsubsection{Determination}
However, we not only want quasi deterministic NFIOAs but deterministic ones because deterministic ones specify finite systems (see appendix and e.g. \cite{Reich2010}).
Fig. \ref{fig_nfioa_overspecify_behavior} shows a simple example, where it seems straight forward to eliminate a spontaneous transition. With the usual interpretation of I/O automata, the two automata of Fig \ref{fig_nfioa_overspecify_behavior} would be viewed as being behaviorally equivalent. 

\begin{figure}[ht]
\begin{center}
\includegraphics[width=10cm]{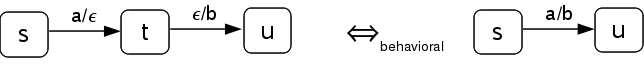} 
\end{center}
\caption[]
{\label{fig_nfioa_overspecify_behavior} \rm Do the NFIOA on the left and on the right specify the same automaton behavior? In case of a channel based restricted product automaton it depends on what the rest of the automaton is doing.}
\end{figure}

But within a channel based restricted product automaton, both behaviors are equivalent only if the rest of the automaton waits until it receives character '$b$' - or, in other words, if any other factor automaton process the result of the first automaton synchronously. With so called asynchronous processing, another factor automaton could transit and show some behavior. Eliminating the intermediate state implies to deprive all other factor automata from their chance to transit, possibly changing the behavior of the entire product automaton.
Eliminating the spontaneous transitions changes the ''atomicity of time steps''. 

Another problem is that the set of state values might change. As the acceptance component directly relates to the set of state values, it must be restated. Also, the projection relation between the restricted product automaton and its constituting roles gets lost. 
A third problem is that we might have to aggregate several output characters in one step. 

\begin{figure}[ht]
\begin{center}
\includegraphics[width=8cm]{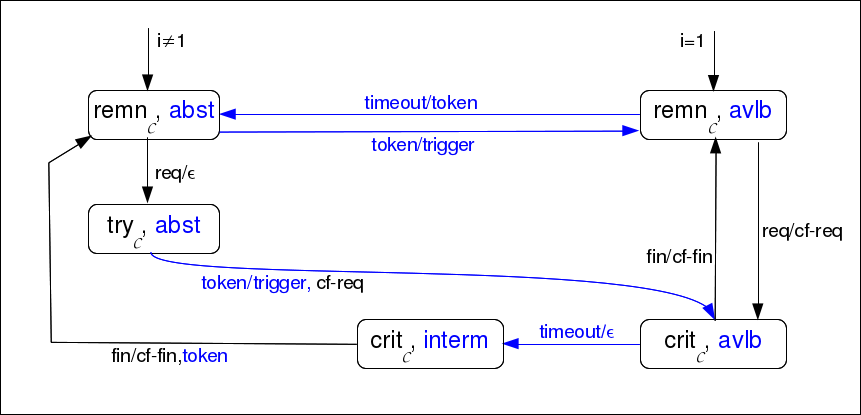}
\end{center}
\caption[]
{\label{fig_mutual_exclusion_token_alg_sync5} The completely deterministic resource administration process.}
\end{figure}

It appears that moving from the quasi deterministic automaton to a deterministic one, the behavioral equivalence of the complete interaction has to be checked. For the quasi deterministic automaton of Fig. \ref{fig_mutual_exclusion_token_alg_sync1a}, representing the resource administration process, a behaviorally equivalent fully deterministic automaton can be found which is illustrated in Fig. \ref{fig_mutual_exclusion_token_alg_sync5}.
So, deterministic processes do coordinate at least two different interactions, but they may have lost their resemblance to the original roles used for the definition of their interactions they are involve in.

%
\section{Related Work} \label{s_related_work}
%
There have been many different approaches to the issue of describing processes and their interactions, many of them can be subsumed under Richard Mayr's \cite{Mayr1999} process rewrite systems. He introduced a unifying view encompassing many formalisms based on named actions like finite state systems, Petri nets, pushdown processes, etc. The most general class he named process rewrite systems.

It presupposes a countably infinite set of atomic actions $Act = \{a, b, c,\dots\}$ and a countably infinite set of process constants $Const = \{\epsilon, X, Y, Z, \dots\}$. All process terms that describe the states of the system have the form $t = \epsilon|X|t_1.t_2|t_1|| t_2$. The dynamics of a system is described by a finite set of rules $\Delta$ of the form $t_1 \stackrel{a}{\rightarrow} t_2$. For every $a\in Act$, the transition relation $\stackrel{a}{\rightarrow}$  is the smallest relation that satisfies the inference rules
\[
\frac{(t_1 \stackrel{a}{\rightarrow} t_2)\in \Delta}{t_1 \stackrel{a}{\rightarrow} t_2}, \quad \frac{t_1 \stackrel{a}{\rightarrow} t_1'}{t_1||t_2 \stackrel{a}{\rightarrow} t_1'||t_2},\quad \frac{t_1 \stackrel{a}{\rightarrow} t_1'}{t_1.t_2 \stackrel{a}{\rightarrow} t_1'.t_2}
\]

How would the first interaction, the protocol of mutual exclusion (see. section \ref{sss_first_interaction}) be described? Instead of denoting the input and output, each transition would have to be arbitrarily named. E.g. for the user process ${\cal U}$: $rem \stackrel{try()}{\rightarrow} try$, $try \stackrel{cf-req()}{\rightarrow} crit$, $crit \stackrel{finalize()}{\rightarrow} exit$, $exit \stackrel{cf-fin()}{\rightarrow} rem$ and for the resource administration process: $rem \stackrel{request()}{\rightarrow} try$, $try \stackrel{send-cf-req()}{\rightarrow} crit$, $crit \stackrel{exit()}{\rightarrow} exit$, $exit \stackrel{send-cf-exit()}{\rightarrow} rem$. 

Because these names or labels are arbitrary, they don't help in constructing the transition relation of the protocol as a restricted product automaton. Here is how the protocol would have to be described:  
\begin{eqnarray*}
rem_{\cal U} || rem_{\cal C} & \stackrel{try()}{\rightarrow} & try_{\cal U}|| rem_{\cal C}\\
try_{\cal U} || rem_{\cal C} & \stackrel{request()}{\rightarrow} & try_{\cal U}|| try_{\cal C}\\
& \dots& 
\end{eqnarray*}
One would not be allowed to mix this with the rules for the single protocol role. As a result, we can not use our knowledge of the parts (the roles) to construct the whole (the interactions). 

In essence all approaches that are subsumed with process rewrite systems assume the ''atomicity of actions'' and provide names for identical actions only because they appear differently under different perspectives. 

Other approaches that are also based on naming of actions suffer from the same shortcomings, for example Robin Milner's calculus of communicating systems \cite{Milner1980,Milner1989}, Charles A. R. Hoare's communicating sequential processes \cite{Hoare1985} or Nancy Lynch's distributed algorithms \cite{Lynch1996}.

The same holds for the currently quite popular BPMN proposal \cite{BPMN_20} to describe processes and their interactions driven by large industry players. There is some similarity between what the BPMN standard calls a private (internal) process and a process in my sense, an abstract (public) process and a role and a choreography or collaboration (global) process and a protocol. However, if a private process or orchestration of services is a combination of executable functionality, it remains unclear, what an 'end-to-end' process should be. The logical difference between a private and a public process is not seen as a projection but some sort of abstraction where only those activities that are used to communicate to other processes plus the order of these activities, are included. A choreography describes the way processes coordinate their interactions, which is claimed to allow the process designer to plan their processes for inter-operation without introducing conflicts, but no formal connection is made in 
the specification.

Also, a comparison to game theory is interesting, The interaction model of game theory is similar to the protocol model. As I have shown in \cite{Reich2009}, a protocol, enriched by a decision alphabet, can be mapped onto a corresponding extensive game without its payoff evaluation. With games, the nondeterminism problem of interaction, that the interactions don't determine the actions, is solved by introducing ''decision'' together with a payoff function. Then the decisions determine the actions and the question how to act becomes an optimization problem. In contrast, my approach to determine single interactions by other interactions leads to a coordination problem.
However, there are already game theoretic methods to describe the simultaneous play of multiple games by single players which naturally are very similar to their process analogons (e.g. \cite{DBLP:conf/clima/GhoshRS10}).

%
\section{Discussion} \label{s_discussion}
%
In this article I presented an approach to describe processes and their interactions with the same technical means, namely nondeterministic input output automata (NFIOAs). I presented two complementary ways to represent the causal relation between the output and the input of different systems, which I called ''outer coupling'', and between the input and output of the same system, which I called ''inner coupling''. Extending \cite{reich2011}, an operator based approach was introduced.

The openness of the interaction networks is represented by the incompleteness of the interaction descriptions in the sense of protocols with respect to the interaction coordinating processes. Constructing a process, one needs to know the rest of the interaction network only up to the next interaction partner's role in the own interactions.

Why do protocols provide so many spontaneous transitions, while it is so difficult to get rid of this spontaneity when constructing deterministic processes? It is this spontaneity which provides the freedom to combine different interactions to coordinating processes. So, protocol design strives for providing role descriptions which can be most easily coordinated with other roles, a property which naturally gets lost in constructing deterministic processes.

Currently, instead of ``protocol'' and ``processes'', often the more fashionable terms ``choreography'' and ``orchestration'' are used, especially in the context of a so called ``service oriented architecture (SOA)'' \cite{SchulteNatis1996}. If we identify ``choreography'' with ``protocol'' and ``orchestration'' with ``process'', we see that a single choreography is not supposed to specify the behavior of a system completely as well as an orchestration is not supposed to specify any ``loose coupling'' between two different applications. On the contrary it becomes understandable why a modeling of processes as stepwise executable activities which is not based on the interaction descriptions leads to strong centralization tendencies and tight coupling of the interacting systems.

The straight forward composition mechanism seems quite attractive for combining different roles flexibly into quasi deterministic processes, just by adding a couple of rules. Unfortunately, things become more complex if we strive for determinism, which we must if we want to execute our processes with the traditional machines.\\[0.5cm]

\noindent
{\bf Acknowledgment} Writing this article would not have been possible without an intense discussion with Hans-Jörg Kreowski beforehand and a daylong return journey by train. I also thank the reviewers for their helpful comments. 

\bibliographystyle{eptcs}
  \bibliography{/home/johannes/bib/theoret_med,/home/johannes/bib/philosophy,/home/johannes/bib/informatics,/home/johannes/bib/develop,/home/johannes/bib/soziologie}

%
\section{Appendix} \label{s_appendix}
%
The importance of finite I/O automata for the theory of processes stems from the fact that DFIOAs can specify finite systems (in a physical sense). Such a system is represented as structure which comprises internal as well as externally accessible states in the sense of time dependent properties together with function which describes the time evolution of the states\footnote{This use of then term 'state' is the reason I used the term 'state values' where normally in automata theory the term 'state' is used.}.

\begin{definition} \label{def_system}
A {\it finite system} is defined by a tuple ${\cal S}=(T, succ, Q, I, O, x, in, out, f)$. $T$ is the enumerable set of time values starting with 0 such that $succ:T\rightarrow T$ is the invertible time successor function. $Q, I$ and $O$ are the finite sets of state values for the internal, input and output states $(x, in, out):T\rightarrow (Q, I^\epsilon, O^\epsilon)$. $f=({f^{ext}}, {f^{int}}):I^\epsilon\times Q\rightarrow O^\epsilon \times Q$ is a  function describing the time evolution or system operation updating the internal and output state in one time step for each $t \in T$:
\[{out(t+1) \choose x(t+1)} = {{f^{ext}}(in(t), x(t)) \choose {f^{int}}(in(t), x(t))}\,.\] 
The state values of the I/O-states are also called {\it characters}. The $n$-fold application of the time successor function $succ$ is written as $t+_{\cal S}n:=succ^n_{\cal S}(t)$. 
\end{definition}

\begin{definition} A DFIOA ${\cal A} = (Q_{\cal A}, I_{\cal A}, O_{\cal A}, q_0, Acc, \Delta)$ {\it specifies} a finite system ${\cal S} = (T, succ, Q_{\cal S}, I_{\cal S},\\ O_{\cal S}, x, in, out, f)$ if $Q_{\cal S} \subseteq Q_{\cal A}$,  $I_{\cal S} \subseteq I_{\cal A}$, $O_{\cal S} \subseteq O_{\cal A}$, $x(0) = q_0$ and for any point in time $t \geq 0$ in every possible sequence, $(x(t), x(t+1), in(t), out(t+1)) \in \Delta_{\cal A}$. 
\end{definition}

\end{document}